\documentclass{article}[11pt]

\usepackage{makeidx}         
\usepackage{graphicx}        
\usepackage{multicol}        
\usepackage[bottom]{footmisc}

\usepackage{amsfonts}
\usepackage{amssymb}
\usepackage{mathrsfs}
\usepackage{amsmath}
\usepackage{amscd}
\usepackage{epsfig}
\usepackage{graphicx,dblfloatfix,caption}
\usepackage{listings}
\usepackage{algorithmic, algorithm}

\newtheorem{theorem}{Theorem}[section]
\newtheorem{lemma}[theorem]{Lemma}

\newenvironment{proof}[1][Proof]{\begin{trivlist}
\item[\hskip \labelsep {\bfseries #1}]}{\end{trivlist}}

\newcommand{\qed}{\nobreak \ifvmode \relax \else
\ifdim\lastskip<1.5em \hskip-\lastskip
\hskip1.5em plus0em minus0.5em \fi \nobreak
\vrule height0.75em width0.5em depth0.25em\fi}

\begin{document}

\title{A Variational Principle for Improving 2D Triangle Meshes based on Hyperbolic Volume}
\author{ Jian Sun 
\thanks{Tsinghua University, Beijing, China \texttt{jsun@math.tsinghua.edu.cn} } 
\and
Wei Chen 
\thanks{Kunming University of Science and Technology, Kunming, China \texttt{chenwei19861027@gmail.com}} 
\and
Junhui Deng 
\thanks{Tsinghua University, Beijing, China \texttt{deng@tsinghua.edu.cn}} 
\and
Jie Gao 
\thanks{Stony Brook University, New York, US \texttt{jgao@cs.sunysb.edu}} 
\and
Xianfeng Gu 
\thanks{Stony Brook University, New York, US \texttt{gu@cs.sunysb.edu}} 
\and
Feng Luo 
\thanks {Rutgers University, New Jersey, US \texttt{fluo@math.rutgers.edu}}
}

%
%
\maketitle

\begin{abstract}
In this paper, we consider the problem of improving 2D triangle meshes tessellating planar regions. We propose a new variational principle for improving 2D triangle
meshes where the energy functional is a convex function over the angle structures whose maximizer is unique and consists only of equilateral triangles. This energy functional is related to hyperbolic volume of ideal 3-simplex. Even with extra constraints on the angles for embedding the mesh into the plane and preserving the boundary, the energy functional remains well-behaved. We devise an efficient algorithm for maximizing the energy functional over these extra constraints. We apply our algorithm to various datasets and compare its performance with that of CVT. The experimental results show that our algorithm produces the meshes with both the angles and the aspect ratios of triangles lying in tighter intervals.
\end{abstract}

\section{Introduction}
\label{sec:intro}
In this paper, we consider the problem of improving 2D triangle meshes tessellating
planar regions. The applications in scientific computing require quality meshes. The quality 
here refers to the shape and size of the elements: a triangle is of good shape if it is close 
to the equilateral triangle, i.e., its inner angles are close to $\pi/3$. 
The most popular approach for generating quality meshes is 
Delaunay refinement~\cite{Chew:1987:CDT}. Delaunay refinement algorithms commonly perform one local
change at a time, until the criteria of the shape and size of elements is satisfied. They are greedy
approaches which may produce bad shaped triangles, especially near the boundary. Quite a few methods 
have been proposed to deal with this issue. Among them, centroidal Voronoi tessellation (CVT) is 
widely used where the vertices are iteratively moved to the barycenters of the corresponding Voronoi 
cells. There is a variational principle associated with this so-call Lloyd iteration where the 
energy functional measures the difference between the site points and the barycenters of the corresponding 
Voronoi cells~\cite{Du:1999:CVT}.

This paper proposes a new variational principle for improving 2D triangle meshes where
the energy functional is a convex function over the angle structures whose maximizer is 
unique and consists only of equilateral triangles. This energy functional is related to 
hyperbolic volume of ideal 3-simplex. Of course, one needs to impose the
extra constraints on the angles to embed the mesh into the plane and preserve the boundary. 
Nevertheless, we show the energy functional is still well-behaved even with these
constraints. 
The space of angle structures is much bigger than the space of coordinates,
which provides more freedom for algorithms to search for maximizer and thus 
find better (local) maximizer and generate better meshes, as demonstrated in Section~\ref{sec:results}. 
We devise an algorithm based on interior-point method for maximizing the energy 
functional over extra constraints. We apply our algorithm to various datasets and compare
its performance with that of CVT. The experimental results show that our algorithm
produces the meshes with both the angles and the aspect ratios of triangles lying in tighter
intervals.


\vspace{2mm}
\noindent{\bf Previous work:~}
There are great amounts of research work on quality mesh generation and many meshing strategies have
been proposed and studied. Here we summarize those most relevant to our work.
Readers are referred to~\cite{del} and the references therein for more related work on quality mesh generation. 
The most popular approach is Delaunay refinement, which iteratively 
inserts Steiner points to improve the quality of the mesh until the initial criteria is satisfied 
for each triangle. Delaunay refinement approach is pioneered by Chew~\cite{Chew:1987:CDT}, and later improved and extended by 
many others\cite{Ruppert:1995:DRA, Shewchuk:2002:DRA}. Shewchuk~\cite{Shewchuk:2002:DRA} shows that it terminates with a finite number of Steiner points and 
with bounds on the angles. Delaunay refinement can be improved either by carefully designing the order 
of inserting Steiner points or choosing the positions of Steiner points other than the 
circumcenters of triangles~\cite{Ungor:2009:ONT}. However, it remains a greedy approach, which may make globally 
bad decisions which are not reversible. To address this issue, variational approaches are proposed where 
an energy functional is chosen so that the low levels of this energy correspond to the meshes with good quality. 
The widely used energy functional is the one used in CVT which sums the differences between the site points
and the barycenters of the corresponding Voronoi cells~\cite{Du:1999:CVT}. This is based on the observation that in 2D, 
evenly distributed points lead to well-shaped triangles in Delaunay triangulation~\cite{eppstein:2001}. Du et al.~\cite{Du:1999:CVT} proposed
the Lloyd iteration to transforms an initial ordinary Voronoi diagram into a centroidal Voronoi diagram. 
Finally, Tournois et al. \cite{Tournois:2009:IDR} proposed to interleave Delaunay refinement and CVT for generating and improving
2D triangle meshes. 

\section{The Energy Functional}
\label{sec:energy}
In this section, we describe the energy functional and discuss its properties. 
We start with a single triangle $t$. Let $\alpha, \beta, \gamma$ be the inner angles
of $t$. Assign $t$ the following energy.
\begin{equation}
E(t) = \Lambda(\alpha) + \Lambda(\beta) + \Lambda(\gamma) 
\end{equation}
where $\Lambda$ is Lobachevsky function: 
\begin{equation}
\Lambda(x) = -\int_0^x \ln |2\sin(t)|dt. 
\label{eqn:lobachevsky}
\end{equation}
Lobachevsky function is continuous odd and periodic of period $\pi$. 
Figure~\ref{fig:lob_fun} shows the graph of $\Lambda$ over $[0, \pi]$. 
See~\cite{Mil} for more properties of Lobachevsky function. 

\begin{figure}[!t]
\begin{center}
\begin{tabular}{c}
\includegraphics[width=0.6\textwidth]{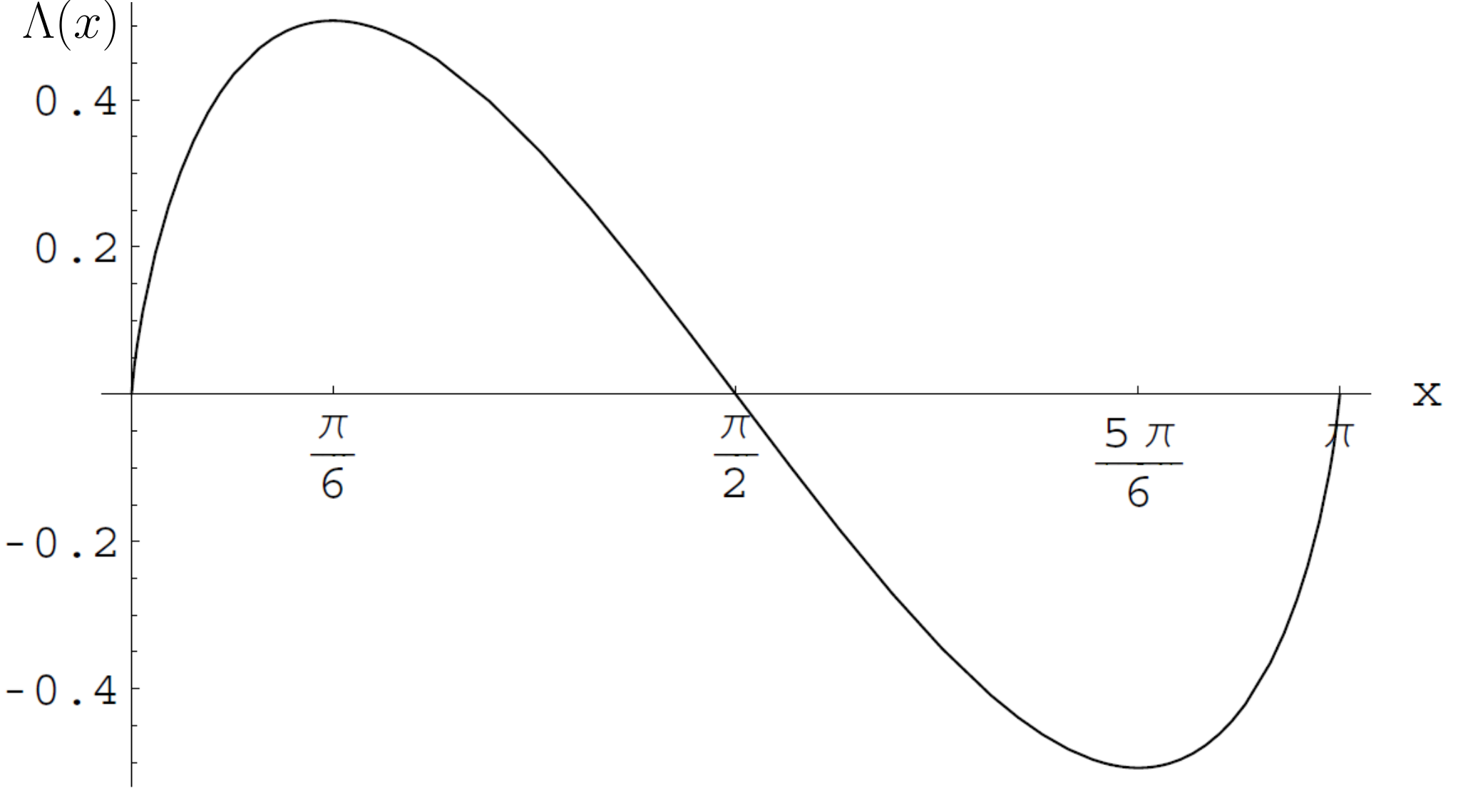}
\end{tabular}
\end{center}
\vspace{-4mm}
\caption{The graph of Lobachevsky function.}
\label{fig:lob_fun}
\end{figure}

This energy assigned to $t$ is in fact the volume of an ideal hyperbolic 3-simplex. Consider
the upper half space: $\mathbb{H}^3 = \{(x, y, z) \in  \mathbb{R}^3| z > 0\}$ with the hyperbolic
metric $ds^2= \frac{dx^2 + dy^2 + dz^2}{z^2}$. Place the triangle $t$ on the plane $z=0$, and add
the fourth vertex $l$ at infinity. Then all four vertices $i, j, k, l$ are at infinity, and 
thus form an ideal hyperbolic 3-simples, denoted $\Delta=ijkl$. Place the vertex $l$ at a position
so that the dihedral angles along three edges meeting at $l$ are the inner angles of the triangle $t$, 
as shown in Figure~\ref{fig:ideal_simplex}. The fact that $E(t)$ is the volume of the ideal 3-simple $\sigma$ follows
from the following lemma.

\begin{lemma}[e.g., \cite{Mil}]
Consider an ideal hyperbolic 3-simplex, that is a simplex $\Delta$ with all four vertices 
at infinity. If $\alpha$, $\beta$, $\gamma$ are the dihedral angles along three 
edges meeting at a common vertex, then $\alpha + \beta + \gamma = \pi$, and 
\begin{equation}
\text{volume}(\Delta) = \Lambda(\alpha) + \Lambda(\beta) + \Lambda(\gamma)
\label{eqn:volume}
\end{equation}
\end{lemma}
Remark that it does not matter which particular vertex we choose, since it follows easily that the 
dihedral angles along the opposite edges of $\Delta$ are equal so that we have the same three 
dihedral angles $\alpha$, $\beta$, $\gamma$ incident to any vertex. 

There are many nice properties of the energy function $E(t)$. Here we state two of them which
are most relevant to our setting.  The following lemma says that $E(t)$ reaches maximum when $t$ is equilateral. 
\begin{lemma}[e.g., \cite{Mil}]
The volume of a hyperbolic 3-simplex reaches the maximum $3\Lambda(\pi/3)$ when $\alpha = \beta =\gamma = \pi / 3$.  
\label{lem:maximum-t}
\end{lemma}

\begin{figure}[!t]
\begin{center}
\begin{tabular}{c}
\includegraphics[width=0.6\textwidth]{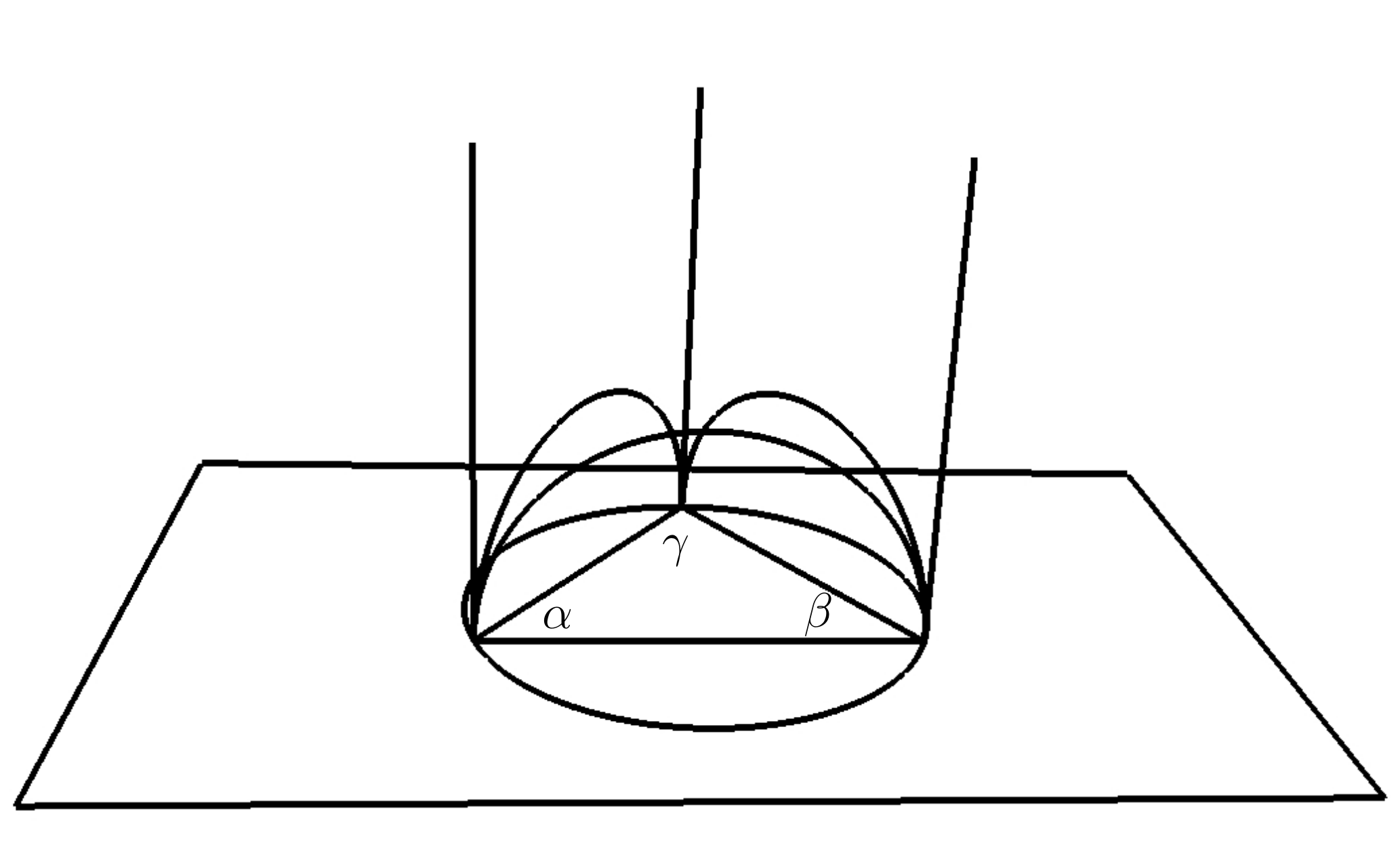}
\end{tabular}
\end{center}
\vspace{-4mm}
\caption{Ideal 3-simplex.}
\label{fig:ideal_simplex}
\end{figure}

Since $\alpha + \beta + \gamma = \pi$, 
$E$ is a function of two inner angles $\alpha$ and $\beta$ which parametrize the space of all 
Euclidean triangles up to similar transformations. The following lemma tells $E$ 
is a strictly concave function over the space of Euclidean triangles up to similar
transformations. 
\begin{lemma}[\cite{r-esssh-94}]
The Hessian of $E$ is 
\begin{equation}
H(E)(\alpha, \beta) = \begin{bmatrix} \frac{-\sin\beta}{\sin(\alpha + \beta)\sin(\alpha)} &  \frac{\cos(\alpha+\beta}{\sin(\alpha + \beta)} \\ 
												  \frac{\cos(\alpha+\beta}{\sin(\alpha + \beta)} &  \frac{-\sin\alpha}{\sin(\alpha + \beta)\sin(\beta)}
													 \end{bmatrix}
\end{equation}
and is negative definite. 
\label{lem:concave-t}
\end{lemma}
\begin{proof}
Based on the equations~(\ref{eqn:lobachevsky}, \ref{eqn:volume}), the Hessian of $E$ can be derived easily. 
To see $H(E)$ is negative definite, let $v = (x, y)^T \in \mathbb{R}^2$ and $\|v\|= 1$, We have
\begin{eqnarray*}
v^t H(E) v &=& - \frac{-\sin\beta}{\sin(\alpha + \beta)\sin(\alpha)} x^2 + \frac{2\cos(\alpha+\beta)}{\sin(\alpha + \beta)}xy+ \frac{-\sin\alpha}{\sin(\alpha + \beta)\sin(\beta)}y^2 \\
&\leq& \frac{-2}{\sin(\alpha + \beta)} xy + \frac{2\cos(\alpha+\beta)}{\sin(\alpha + \beta)}xy < 0, 
\end{eqnarray*}
since $\sin\alpha > 0, \sin\beta > 0, \sin(\alpha + \beta) > 0, |\cos(\alpha+\beta)| < 1$
\end{proof}

Now we are ready to define the energy functional for a triangle mesh $T$.  
Denote the sets of vertices, edges and triangles of $T$ by $V$, $E$ and $F$. 
We identify a vertex in $T$ with an index, i.e., $V = \{1, 2, \cdots, n\}$, where $n$ is 
the number of vertices in $T$. We denote by $ij$ the edge with vertices $i$ and $j$, by $ijk$, 
the triangle with vertices $i$, $j$ and $k$.  The triangles in $T$ are Euclidean.
Let $\alpha_{jk}^i$ denote the inner angle at vertex $i$ in triangle $ijk$. Let $|X|$ 
denote the cardinality of a set $X$.
Define the energy functional as the sum of the energy over all 
triangles in $T$, i.e.,  
\begin{equation}
\mathcal{E}(T) = \sum_{t\in F} E(t). 
\end{equation}

Let ${\bf A}_T$ denote all possible angle structures given the combinatorial 
structure of $T$, i.e., 
\begin{eqnarray*}
{\bf A}_T= \{(\cdots, \alpha_{jk}^i, \alpha_{ki}^j, \alpha_{ij}^k, \cdots)^t \in 
\mathbb{R}^{3|F|}| 
\text{~for all~} &&ijk \in F: \alpha_{jk}^i + \alpha_{ki}^j+ \alpha_{ij}^k = \pi,\\
									&&\alpha_{jk}^i > 0, \alpha_{ki}^j>0, \alpha_{ij}^k>0 \}. 
\end{eqnarray*}
If the combinatorial structure of $T$ is fixed, the energy $\mathcal{E}$ is a function 
over the angle structures ${\bf A}_T$. 
The following lemma follows easily from Lemma~\ref{lem:concave-t}.
\begin{lemma}
$\mathcal{E}$ is a concave function over the angle structures ${\bf A}_T$. 
\end{lemma}
We remark that  (1) From Lemma~\ref{lem:maximum-t}, $\mathcal{E}$ reaches the
maximum over ${\bf A}_T$ when all triangles in $T$ become equilateral, 
(2) In the meshing application considered in this paper, additional constraints are necessary 
to impose on the angles as we will discuss in the next section. Nevertheless, the triangles 
in the mesh $T$ become more well-shaped (closer to equilateral) when the energy $\mathcal{E}$
increases. 

\section{ Angle Structures and Embeddings}
\label{sec:anglestructure}
In the paper, we consider the problem of improving a triangle mesh $T$ which tessellates 
a planar region. Therefore the triangle mesh $T$ is embedded in the plane where each vertex
$i$ has a coordinate $(x_i, y_i)$ in the plane. The quality of triangles are improved by 
adjusting the coordinates of vertices. On the other hand, our energy functional is defined 
over the angles. In our method, the coordinates of the vertices are adjusted by changing
the angles. Therefore, it is necessary to relate the embeddings of $T$ and its angle structures. 
We relate them using the metric. The metric of a triangle mesh $T$ specifies the 
length for each edge in the mesh, or equivalently, is a function $d: E\rightarrow \mathbb{R}_{>0}$ 
such that the triangle inequalities hold for each triangle, i.e., $d(ij) + d(jk) > d(ik), d(ij) + d(ik) > d(jk), 
d(ik) + d(jk) > d(ij)$ for triangle $ijk$. 

It is obvious that an embedding of a triangle mesh, or equivalently, the coordinates
$\{(x_i, y_i)\}_{i\in V}$, induces a metric for the triangle mesh where $d(ij) = \|(x_i, y_i) - (x_j, y_j)\|$, 
for any $ij \in E$, and a metric $d$ of a triangle mesh induces an angle structure $A \in {\bf A}_T$ 
where $\cos\alpha_{ij}^k =\frac{d^2(ki) + d^2(kj) - d^2(ij)}{2d(ki)d(kj)}$ 
for any $ijk \in F$. 
Furthermore, the induced angle structure is an invariant of rigid transformations 
(translations, rotations and reflections) and uniform scaling of the embedding of the triangle mesh. 
In fact, one can recover the embedding from the induced angle structure as follows. Observe that 
if the inner angles of a triangle are given, one can calculate the coordinate of the third vertex 
from the coordinates of the other two. First, pick a triangle in $F$ and embed it into the plane 
with given inner angles, i.e., compute the coordinates for its three vertices. Up to a rigid 
transformation and a uniform scaling, these coordinates are uniquely determined. Then consider 
its neighboring triangles, each of which has two of its 
vertices already embedded into the plane. Based on the previous observation, the coordinate of 
its third vertex is uniquely determined from the angles and can be easily computed. Once embed these 
neighboring triangles in the plane, consider their neighboring triangles and repeat the above 
procedure until all triangles get embedded into the plane. This leads to an algorithm to layout 
a triangle mesh based on its angle structure. See Algorithm 3. 

\begin{figure}[t]
\begin{center}
\begin{tabular}{cc}
\includegraphics[width=0.45\textwidth]{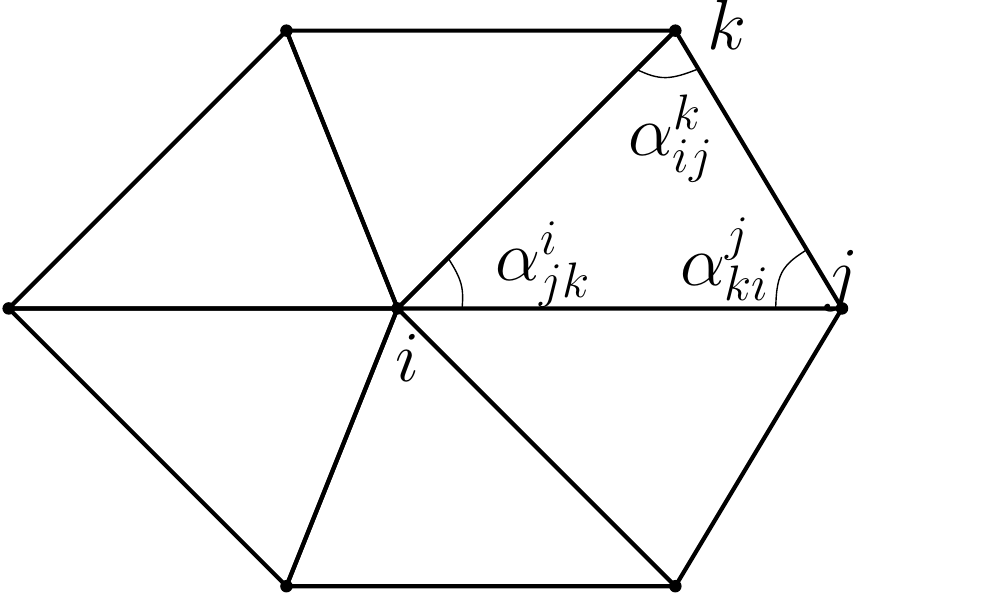} & \includegraphics[width=0.45\textwidth]{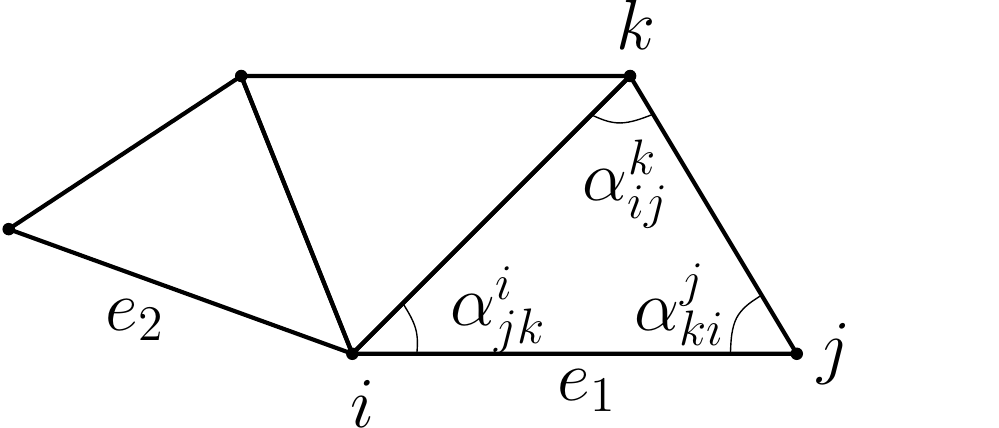} \\
(a) & (b)
\end{tabular}
\end{center}
\vspace{-4mm}
\caption{Constraints on angles. }
\label{fig:one_ring}
\end{figure}

On the other hand,  not every angle structure $A\in \bf{A}_T$ is induced from an embedding of $T$ 
on the plane. To see what are the constraints that the angle structure
induced from an embedded triangle mesh satisfies, consider the one-ring neighbor of a vertex either in the interior 
or on the boundary, as shown in Figure~\ref{fig:one_ring}. We use the following two quantities
to describe the constraints. One quantity is the the angle sum at vertex $i$ given
an angle structure $A$:
\begin{equation}
\Theta(i, A) = \sum_{jk: ijk\in F} \alpha_{jk}^i.
\label{eqn:anglesum}
\end{equation}
The other quantity is the so-called holonomy at vertex $i$ given an angle structure $A$
\footnote{Precisely, it is the holonomy of a sequence of the triangles 
incident to vertex $i$. See~\cite{r-esssh-94}}:
\begin{equation}
H(i, A) = \sum_{jk: ijk \in F} (\ln\sin\alpha_{ij}^k - \ln\sin \alpha_{ki}^j).
\label{eqn:holonomy}
\end{equation}

Consider an interior vertex $i$ of a triangle mesh on the plane. 
See Figure~\ref{fig:one_ring}(a).  
First, the angle sum at  interior vertex $i$ has to be $2\pi$, i.e., 
\begin{equation}
\Theta(i, A)  = 2\pi.
\label{eqn:anglesum-i}
\end{equation}
Second, notice that if let $l_{ij}^k$ be the length of the edge opposite to vertex $k$ in triangle $ijk$, 
then  $\prod_{jk: ijk \in F} l_{ij}^k/l_{ki}^j = 1$. By law of sines, we have for any interior vertex $i$ 
\begin{equation}
H(i, A) = 0.
\label{eqn:holonomy-i}
\end{equation}
One can show that if its angle structure satisfies equation~(\ref{eqn:anglesum-i}, \ref{eqn:holonomy-i}) for each 
interior vertex $i$, the triangle mesh $T$ is locally flat and can then be immersed into the plane. An immersion
is locally a one-to-one map. However, note that an immersion is not 
necessary an embedding, which may have global self-intersections~\cite{dt}. It is relatively hard to 
impose or even describe the constraints on the angles to circumvent global self-intersections. 
Fortunately, we have not observed such global self-intersections in our experiments. This may be due to 
the fact that we also preserve the boundary as described below, which makes them very rare. 

In addition, as our purpose is to improve the quality of a mesh tessellating a fixed planar region, we want to 
preserve the boundary. For simplicity, assume the boundary has one connected component. 
See section~\ref{sec:algorithm} for how we deal with multiple connected components on the boundary. 
Consider a vertex $i$ on the boundary. Let $e_1$ and $e_2$ are two edges on the boundary incident to vertex $i$. 
See Figure~\ref{fig:one_ring}(b). 
First, the angle sum at a boundary vertex $i$ need to be preserved, i.e., 
\begin{equation}
\Theta(i, A) = \Theta_i
\label{eqn:anglesum-b}
\end{equation}
where $\Theta_i$ is the angle between $e_1$ and $e_2$ containing the interior of the planar region.  
Second, the ratio between the length of two consecutive edges on the boundary need to be preserved. One can
write this ratio in terms of the holonomy of the boundary vertex $i$  
\begin{equation}
H(i, A) = \ln (l_1/l_2)
\label{eqn:holonomy-b}
\end{equation}
where $l_1, l_2$ are the length of $e_1, e_2$ respectively. 
Note if the angle structure satisfies equation~(\ref{eqn:anglesum-b}, \ref{eqn:holonomy-b}) 
for each boundary vertex $i$, then the boundary is preserved, up to a rigid transformation 
and a uniform scaling.  

In summary, we have two types of extra constraints imposed on the angle structures so that the triangle mesh $T$
can be immersed into the plane with the shape of the boundary preserved. One type is the angle sum 
imposed on each vertex. Given $\Theta \in \mathbb{R}^{|V|}$ with $\Theta_i$ specifying the angle 
sum at vertex $i$, define a subset of ${\bf A}_T$ as 
\begin{eqnarray*}
{\bf L}_{T, \Theta}= \{ A \in {\bf A}_T| 
\text{~for all~} &&i \in V: \Theta(i, A)  = \Theta_i\}. 
\end{eqnarray*}
The other type is the holonomy condition imposed on each vertex. 
Given $H \in \mathbb{R}^{|V|}$ with $H_i$ specifying the holonomy at vertex $i$, define
another subset of ${\bf A}_T$ as
\begin{eqnarray*}
{\bf N}_{T, H}= \{ A \in {\bf A}_T| 
\text{~for all~} &&i \in V: H(i, A)  = H_i\}. 
\end{eqnarray*}
Since the total sum of all the inner angles has to equal $\pi |F|$, there are only $|V|-1$ number of 
independent equality constraints in defining the subset ${\bf L}_{T, \Theta}$. In addition, the  
sum off the holonomy of all the vertices has to be $0$ since each inner angle appears twice in this 
sum, once positive and once negative. So there are also $|V|-1$ number of  
independent equality constraints in defining the subset ${\bf N}_{T, H}$. 

Observe that the equality constraints on the angle sum is linear in angles and thus
${\bf L}_{T, \Theta}$ is always a convex subset of ${\bf A}_{T}$, which leads to
the following lemma. 
\begin{lemma}
$\mathcal{E}$ is a concave function over the subset of the angle structures ${\bf L}_{T, \Theta}$
for any $\Theta$. 
Furthermore, if $A \in {\bf L}_{T, \Theta}$ is an extremal point of $\mathcal{E}$, then the holonomy at each
interior vertex is automatically $0$. 
\label{lem:concave-lt}
\end{lemma}

The proof for an interior vertex having $0$ holonomy given an extremal point of $\mathcal{E}$ 
can be find for example in~\cite{dc}. However, in order to preserve the shape of the boundary, 
one need to impose the extra holonomy conditions to all boundary vertices, which unfortunately
does not hold automatically. These equality constraints on the holonomy are nonlinear in angles. 
Thus the maximization of the energy functional $\mathcal{E}$ becomes non-convex. 

\section{The Algorithms}
\label{sec:algorithm}
In this section, we describe an algorithm which takes input a triangle mesh tessellating
a region on the plane, and outputs another triangle mesh tessellating the same region with
the shape of the triangles improved to closer to the equilateral triangles. We sketch our
{\bf remeshing} algorithm in pseudo-code as follows: 
\begin{algorithm}[h!]
\floatname{algorithm}{Algorithm}
\caption{ { {\bf remeshing}($T_0=(V_0, E_0, F_0), \{(x_i, y_i)\}_{i=1}^{|V_0|}$) } }
\begin{algorithmic}[1]
\STATE Call $\text{\bf cut}(T_0, \{(x_i, y_i)\}_{i=1}^{|V_0|})$ to cut the mesh $T_0$ into a topological 
disk to connect the different connected components of the boundary and obtain a new mesh $T=(V, E, F)$. 
\STATE Compute the angle structure $A$ of $T$ induced by the coordinates $\{(x_i, y_i)\}_{i=1}^{|V|}$.
\STATE Set $\Theta\in \mathbb{R}^{|V|}$ so that $\Theta_i = \Theta(i, A)$. 
\STATE Set $H \in \mathbb{R}^{|V|}$ so that $H_i=H(i, A)$. 
\STATE Call $\text{{\bf argmax}}(T, A, \Theta, H)$ to maximize the energy $E$ over ${\bf L}_{T, \Theta} \cap {\bf N}_{T, H}$ 
and obtain the corresponding angle structure, denoted $A^*$. (Section~\ref{}). 
\STATE Call $\text{{\bf layout}}(T, A^*, \{(x_i, y_i)\}_{i=1}^{|V|})$ to layout the triangle mesh $T$ based on $A^*$ and 
compute the new coordinates $\{(x^*_i, y^*_i)\}_{i=1}^{|V|}$ for vertices. (Section~\ref{}).  
\STATE Output the new triangle mesh $T_0 = (V_0, E_0, F_0), \{(x^*_i, y^*_i)\}_{i=1}^{|V_0|}$. 
\end{algorithmic}
\end{algorithm}

When there are multiple connected components on the boundary, the constraints described in 
Section~\ref{sec:anglestructure} only preserve the shape for each component but not their
relative position. To deal with this issue, in the first step, we connect the different components 
using the shortest paths and then cut the mesh along them into a topological disk with only one boundary
component. In this step, any vertex $i$ on a cutting path may be split into 
several copies in $T$ each of which takes $(x_i, y_i)$ as its initial coordinates. In the final step
of outputting the new mesh, these copies of vertex $i$ in $T$ have the same new coordinates as the 
shape of the boundary is preserved. So just take one of them as the new coordinate $(x^*_i, y^*_i)$ for
vertex $i$ in $T_0$.  Figure~\ref{fig:hole3_alg} illustrates the algorithm. As we can see, that our algorithm
improves the quality of the triangles, especially those near the boundary. 

\begin{figure}[!t]
\begin{center}
\begin{tabular}{ccc}
\includegraphics[width=0.25\textwidth]{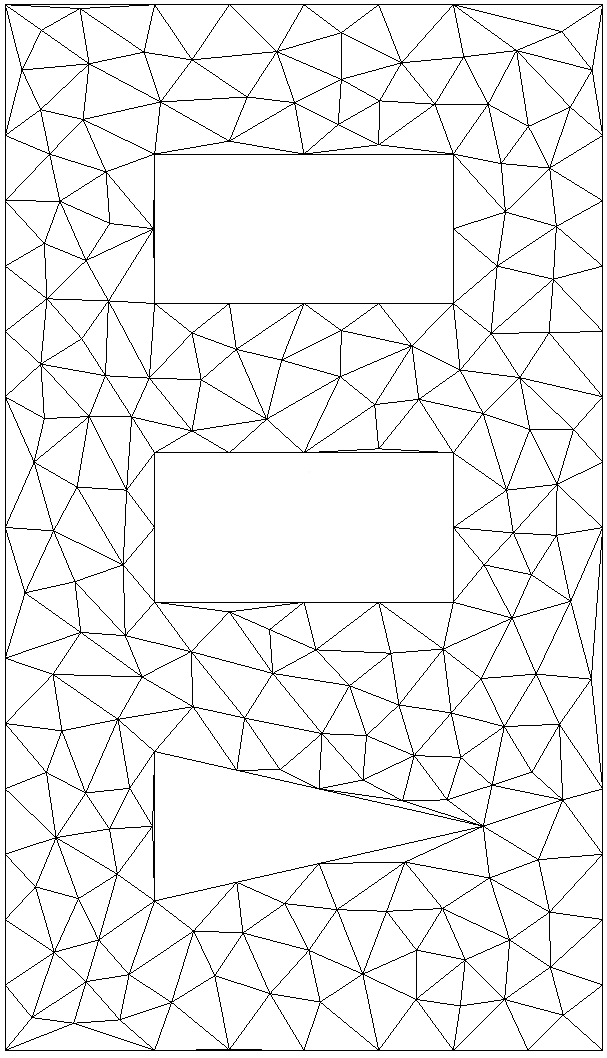} & \includegraphics[width=0.25\textwidth]{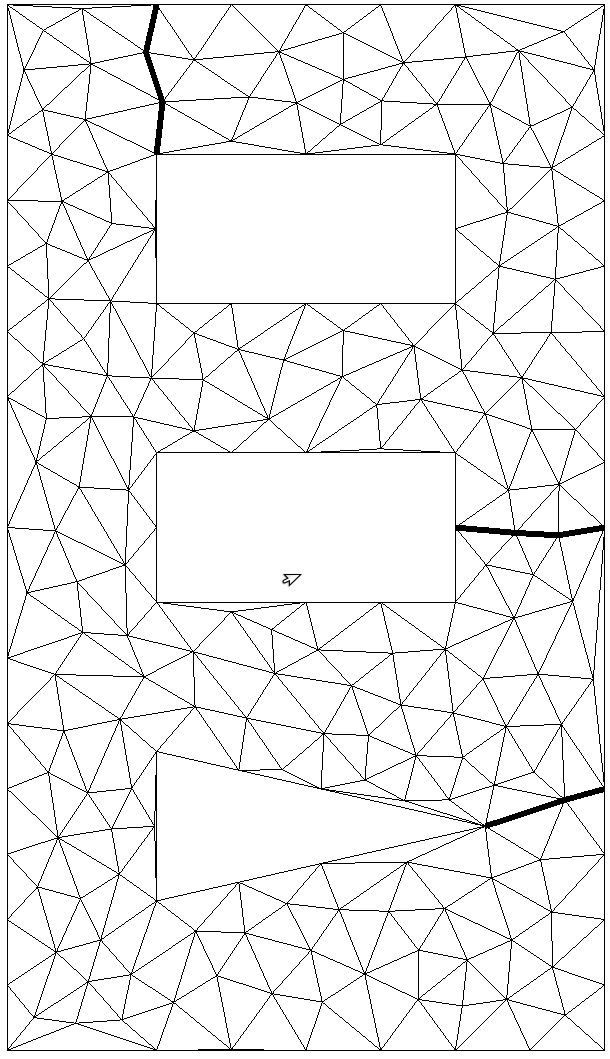} & \includegraphics[width=0.25\textwidth]{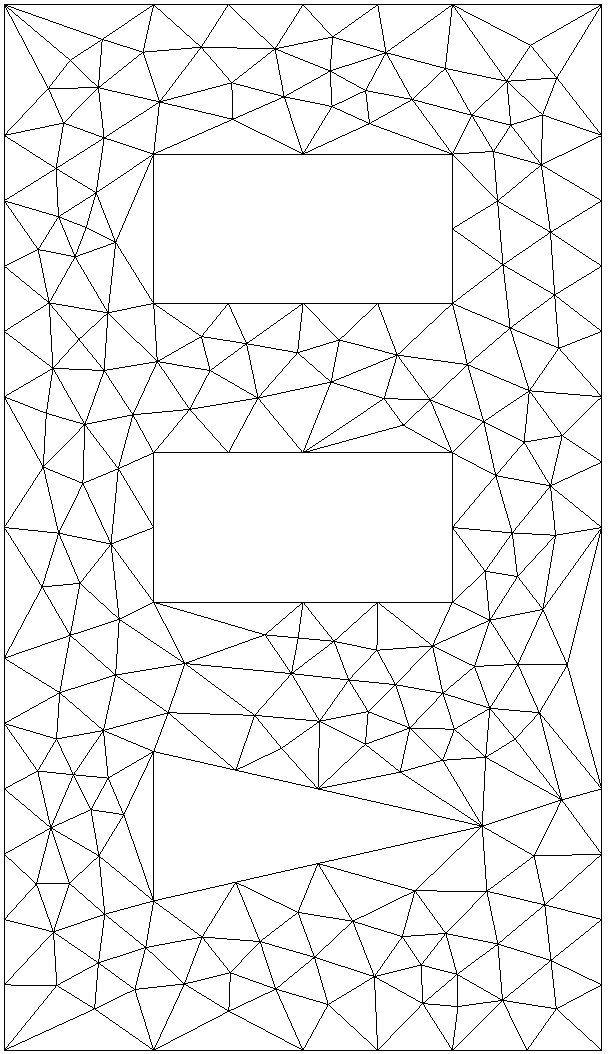}\\
input mesh & cut paths & output mesh
\end{tabular}
\end{center}
\vspace{-4mm}
\caption{The bold edges in the middle picture are the paths along which the mesh is cut into a topological disk.  
}
\label{fig:hole3_alg}
\end{figure}

\subsection{Maximize $\mathcal{E}$}
In this subsection, we describe an algorithm to find an angle structure $A^*$ which maximizes 
the energy functional $\mathcal{E}$ over the subset ${\bf L}_{T, \Theta} \cap {\bf N}_{T, H}$ of angle structures. 
This is the key part of the algorithm. 
The subset ${\bf N}_{T, H}$ is nonlinear in angles. Thus it is an optimization
problem over nonlinear constraints. We use the interior-point method and follow the implementation of
the Matlab routing {\emph{fmincon}} to solve this optimization problem. 

Since there are $|V|-1$ number of nonlinear 
equality constraints in defining the subset ${\bf N}_{T, H}$, it is hard for the interior-point method 
to search for feasible solutions in a subset of this high codimension. To address this issue, we break
the optimization procedure into two steps. 
In the first step. we maximize $\mathcal{E}$ over ${\bf L}_{T, \Theta}$. 
By Lemma~\ref{lem:concave-lt}, $\mathcal{E}$ is concave over ${\bf L}_{T, \Theta}$ and has a unique maximum. 
This step can be done very efficiently using the routing $fmincon$.  In the second step, 
we minimize another energy functional $\mathcal{D}$ over ${\bf L}_{T, \Theta}$ where given a vector $H \in \mathbb{R}^{|V|}$, 
\begin{equation}
\mathcal{D}(A) = \sum_{i\in V}  (H(i, A) - H_i)^2, 
\end{equation}
for any $A \in {\bf A}_T$. 
$\mathcal{D}$ measures how much the angle structure $A$ violates the nonlinear equality constraints, 
and reaches the minimum $0$ when $A$ is in ${\bf N}_{T, H}$.  This minimization is the most time consuming 
step of the algorithm.  In the implementation, we supply the gradient and the Hessian matrix of both energy 
functional $\mathcal{E}$ and $\mathcal{D}$ to the routing {\emph{fmincon}}, which significantly improves the 
efficiency of the algorithm. 
The pseudo-code of the algorithm {\bf argmax} is follows: 
\begin{algorithm}[h!]
\floatname{algorithm}{Algorithm}
\caption{ { {\bf argmax}($T = (V, E, F), A, \Theta, H$) } }
\begin{algorithmic}[1]
\STATE Maximize $\mathcal{E}$ over ${\bf L}_{T, \Theta}$ 
using {\emph{fmincon}} with initial guess $A$ and obtain a new angle structure $A_1$. 
\STATE Minimize $\mathcal{D}$ over ${\bf L}_{T, \Theta}$  using {\emph{fmincon}} with initial guess $A_1$ and obtain a new angle structure $A_2$. 
\STATE Output the angle structure $A^*=A_2$. 
\end{algorithmic}
\end{algorithm}

Note that in Algorithm 2, the second step of minimizing $\mathcal{D}$ is necessary for preserving the boundary
but may decrease $\mathcal{E}$ and deteriorate the quality of triangles. In fact, the initial angle structure $A$ 
is a minimizer of $\mathcal{D}$. To see the performance of this step, we tested the following procedure to maximize $\mathcal{E}$ 
over ${\bf L}_{T, \Theta} \cap {\bf N}_{T, H}$. For any $\delta > 0$, define
\begin{eqnarray}
{\bf N}_{T, H, \delta}= \{ A \in {\bf A}_T| \mathcal{D}(A) < \delta \}. 
\end{eqnarray}
Notice that ${\bf N}_{T, H, 0} = {\bf N}_{T, H}$. 
We relax the constraints of nonlinear equality to inequalities and then maximize $\mathcal{E}$ over 
an enlarged subset ${\bf L}_{T, \Theta} \cap {\bf N}_{T, H, \delta}$ for some $\delta$, 
and then reduce $\mathcal{D}$ within ${\bf L}_{T, \Theta}$ along its negative gradient to $\delta /4$. 
Repeat the above steps with $\delta/2$. In this way, we make sure the final $A^*$ is a local maximum of
$\mathcal{E}$ in ${\bf L}_{T, \Theta} \cap {\bf N}_{T, H}$. 
We observe that this procedure produces a mesh with the same quality as  {\bf argmax} but needs 
significantly more computation time. This shows that the second step of minimizing $\mathcal{D}$ 
somehow also respects the energy functional $\mathcal{E}$ well, which is worth further investigation.

\subsection{Layout Mesh}
The procedure of layout a triangle mesh from its angle structure is described in Section~\ref{sec:anglestructure}. 
Since the shape of the boundary is preserved, we embed the vertices on the boundary using their original 
coordinates. Now for any triangle with one edge on the boundary, the coordinates of the third vertex 
can then be computed uniquely from the angles. Then propagate and compute the coordinates for the other vertices.
We sketch the algorithm {\bf layout} in pseudo-code as follows:
\begin{algorithm}[h!]
\floatname{algorithm}{Algorithm}
\caption{ { {\bf layout}($T = (V, E, F), A^*, \{(x_i, y_i)\}_{i=1}^{|V|}$) } }
\begin{algorithmic}[1]
\STATE Embed the vertices on the boundary using their original coordinates $\{(x_i, y_i)\}$
\STATE Compute the new coordinates $(x^*_i, y^*_i)$ of the third vertex for each triangle with at least one edge on the boundary and mark it. 
\STATE Push their neighboring triangles into queue $Q$ and mark them too. 
\REPEAT
\STATE Pop a triangle from $Q$ and compute the new coordinates $(x^*_i, y^*_i)$ for the third vertex. 
\STATE For each neighboring triangle, if it is not marked, push it into $Q$ 
\UNTIL{$Q$ is empty}
\STATE Output the new coordinates $\{(x^*_i, y^*_i)\}_{i=1}^{|V|}$ 
\end{algorithmic}
\end{algorithm}

\subsection{Cut Mesh}
In this subsection, we describe a method to implement the cut algorithm. The basic idea is 
to find some paths to connect the different connected components together and then cut
the mesh along these paths. The pseudo-code of algorithm {\bf cut} is as follows: 
\begin{algorithm}[h!]
\floatname{algorithm}{Algorithm}
\caption{ { {\bf cut}($T_0 = (V_0, E_0, F_0), \{(x_i, y_i)\}_{i=1}^{|V|}$) } }
\begin{algorithmic}[1]
\STATE Obtain the connected components of the boundary, denoted $\{B_1, \cdots, B_k\}$.
\STATE Compute the weight $w_{ij} = \|(x_i, y_i) - (x_j, y_j)\|$ for edge $ij$. 
\STATE Run Dijkstra's algorithm on the weighted $1$-skeleton of $T_0$ with all vertices on the boundary as sources,
and mark a vertex as $i$ if $d(i, B_i) \leq d(i, B_j)$. 
\STATE Construct a graph $G$ where nodes are $\{B_i\}$, and $B_i, B_j$ are connected if there is an edge in 
$E_0$ whose endpoints are marked as $i, j$. 
\STATE For each edge $B_iB_j$ in $G$, compute the shortest path connecting $B_i$ and $B_j$ and use its length to 
weigh the edge $B_iB_j$ in $G$. 
\STATE Compute the minimal spanning tree $T_G$ of $G$, and cut the mesh along the shortest path connecting $B_i$ and $B_j$
if edge $B_iB_j$ is in $T_G$. 
\STATE Output the cut mesh $T=(V, E, F)$.  
\end{algorithmic}
\end{algorithm}

\vspace{2mm}
\noindent{\bf Complexity:} The complexity of the algorithm is dominated by optimization step.  
The complexity of algorithm {\bf layout} and algorithm {\bf cut} are $O(|F|)$ and  $O(|E|+|V|\log|V|)$ respectively.  
See section~\ref{sec:results} for the performance of the algorithm over various datasets.

\section{Results}
\label{sec:results}
In this section, we apply our meshing algorithm to various datasets and show its performance. All input meshes
are obtained by the Delaunay refinement algorithm \emph{triangle} by Shewchuk~\cite{Shewchuk:2002:DRA}.  

Figure~\ref{fig:gear_01} and~\ref{fig:lock_1566} show two meshes and compare our algorithm with the standard centroidal Voronoi tessellation (CVT)
where the boundary is fixed. Following~\cite{Tournois:2009:IDR}, we measure the quality of elements based on its inner angles and 
aspect ratios (circumradii to shortest edge ratios). As we can see, our method performs better in $L_\infty$ sense
while CVT may do better in $L_2$ sense. Namely, both the angles and the aspect ratios in the meshes generated by our 
method lie in tighter intervals than CVT, while the meshes generated by CVT may have more triangles close to the equilateral
one. 

\begin{figure}[!h]
\begin{center}
\begin{tabular}{ccc}
\includegraphics[width=0.3\textwidth]{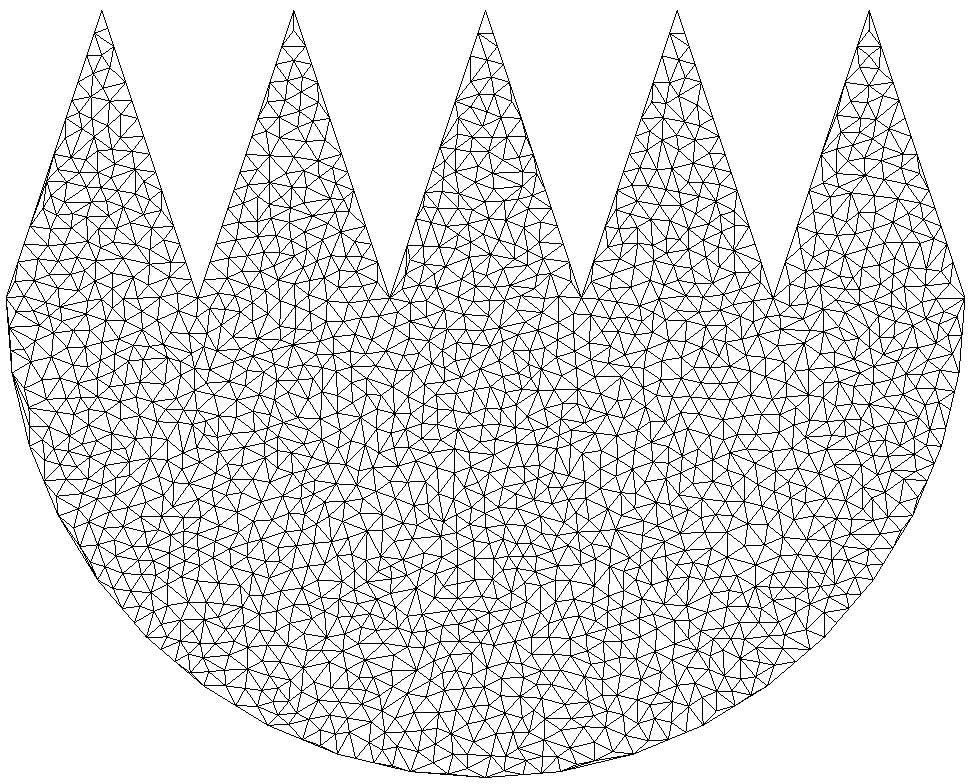} & \includegraphics[width=0.3\textwidth]{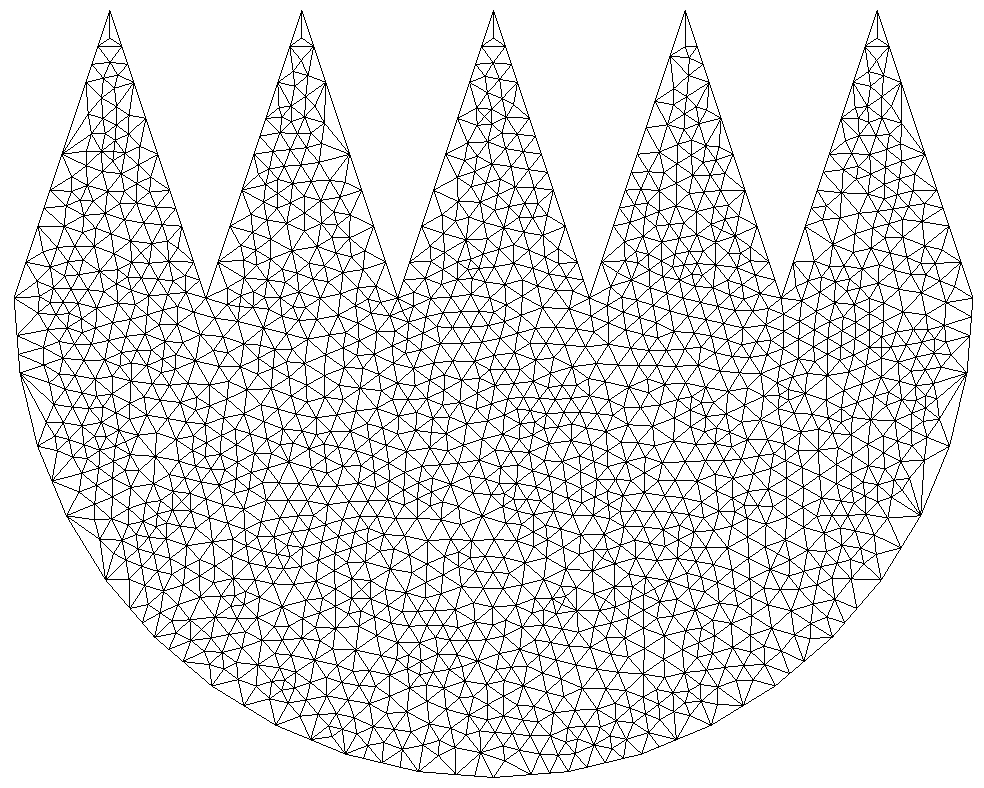} & \includegraphics[width=0.3\textwidth]{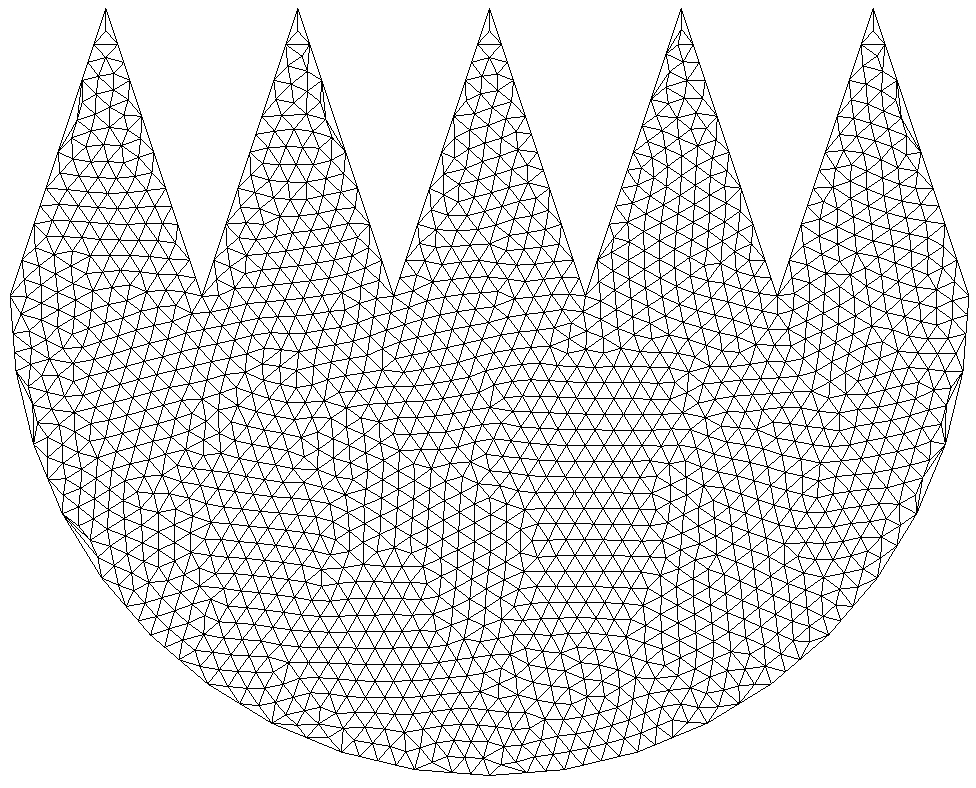}\\
Delaunay refinement & Our method & CVT
\end{tabular}
\begin{tabular}{cc}
\includegraphics[width=0.45\textwidth]{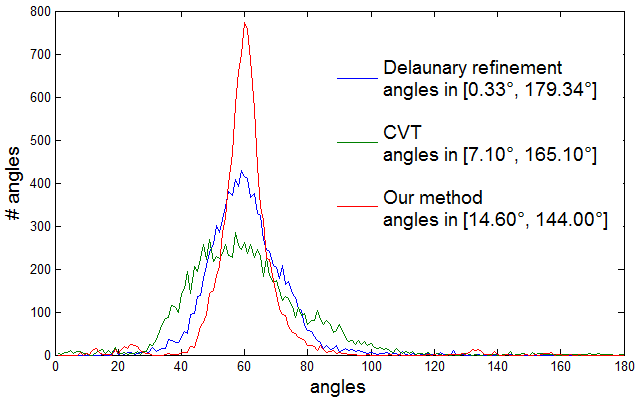} & \includegraphics[width=0.45\textwidth]{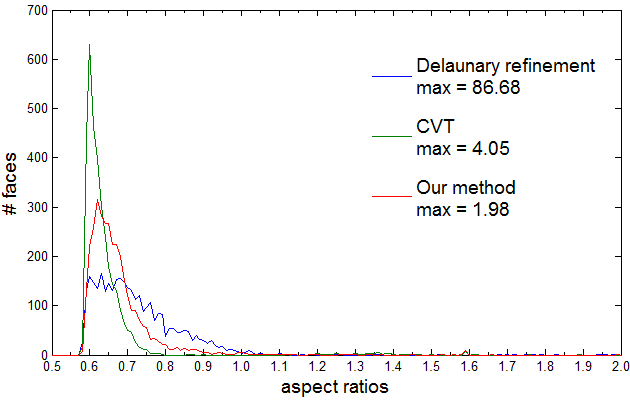}
\end{tabular}
\end{center}
\vspace{-4mm}
\caption{The first row shows the meshes and the second row shows the histograms of the angles and the aspect ratios respectively. 
}
\label{fig:gear_01}
\end{figure}

\begin{figure}[!h]
\begin{center}
\begin{tabular}{ccc}
\includegraphics[width=0.3\textwidth]{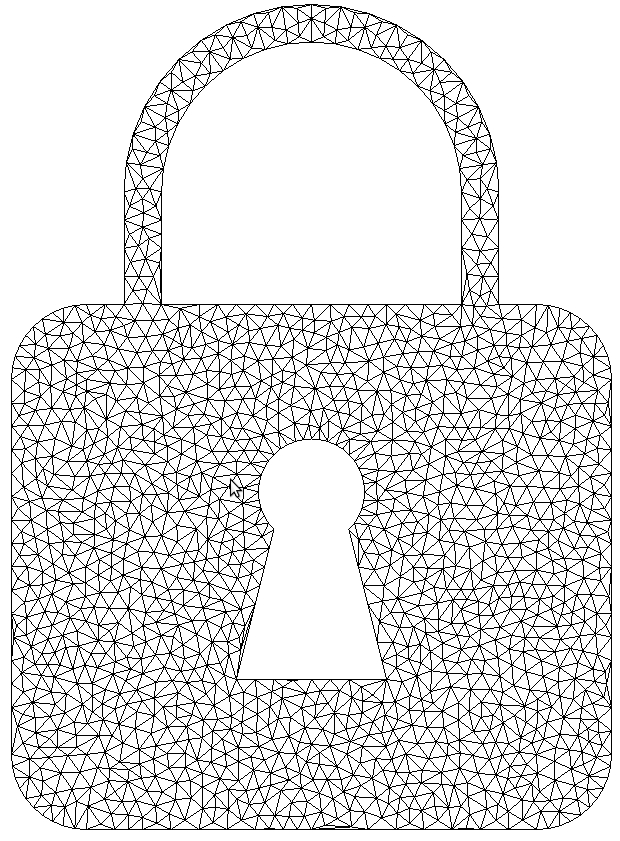} & \includegraphics[width=0.3\textwidth]{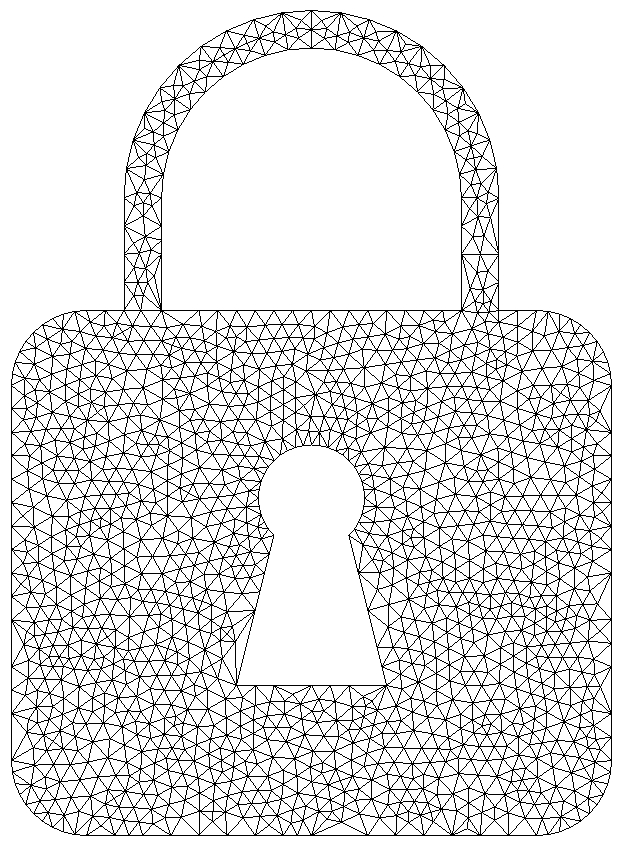} & \includegraphics[width=0.3\textwidth]{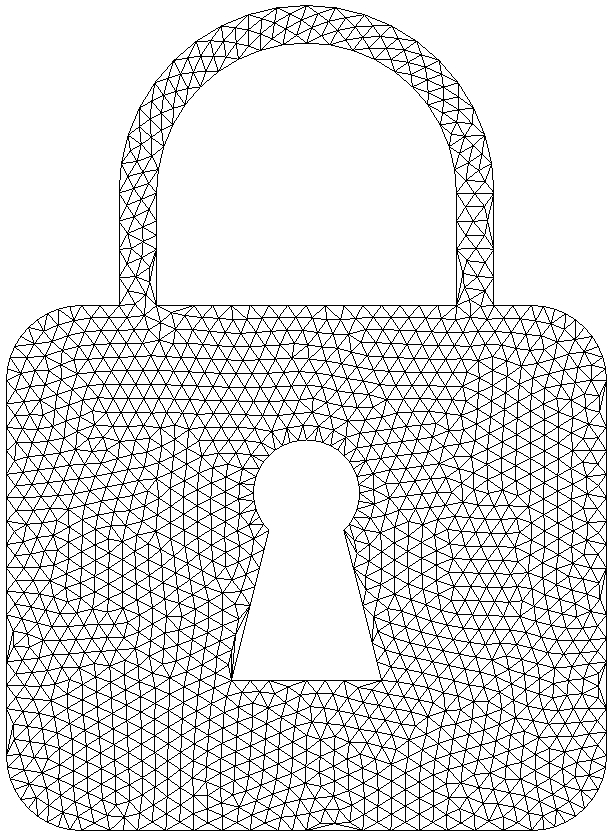}\\
Delaunay refinement & Our method & CVT
\end{tabular}
\begin{tabular}{cc}
\includegraphics[width=0.45\textwidth]{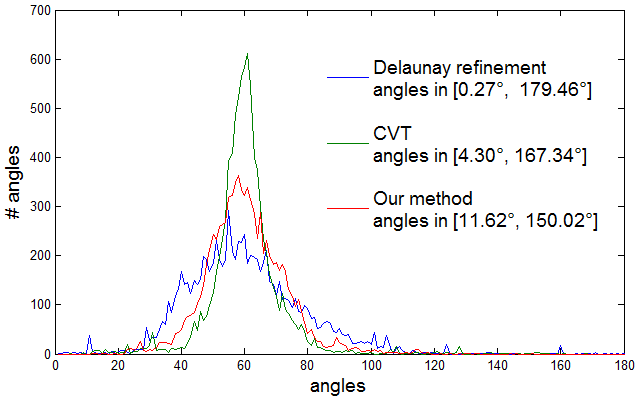} & \includegraphics[width=0.45\textwidth]{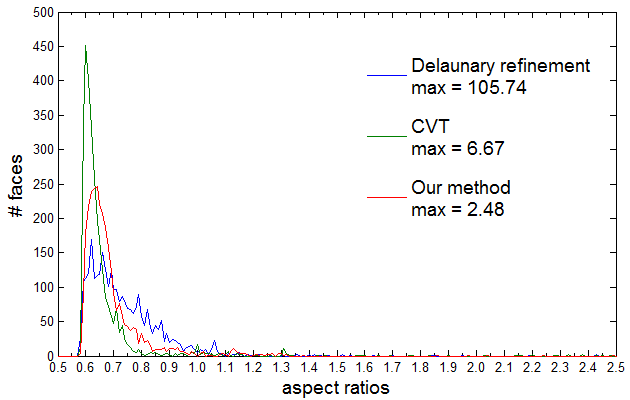} 
\end{tabular}
\end{center}
\vspace{-4mm}
\caption{The first row shows the meshes and the second row shows the histograms of the angles and the aspect ratios respectively. 
}
\label{fig:lock_1566}
\end{figure}

Table~\ref{tbl:timing} shows the timings of the procedure {\bf argmax} over input meshes with 
increasing number of vertices.  The other procedures of the algorithm are negligible in terms
of timing. Figure~\ref{fig:H_hole3} shows the meshes with three different resolutions for each model. 
As we can see, the procedure of minimizing $\mathcal{D}$ consumes most of computation time, 
which also increases faster as the number of vertices increases. 

Finally Figure~\ref{fig:family} and~\ref{fig:taj} show the meshes generated by our algorithm on some interesting 
planar regions. 

\begin{table*}[h]
\begin{center}
\begin{tabular}{|c|c|c|c|c|c|}
\hline
$|V|$ in model H &~~~393~~~&~~~756~~~&~~~1476~~~&~~~3584~~~&~~~7072~~~\\
\hline
Max $\mathcal{E}$ &  5 & 23 & 61 & 249 & 855 \\
\hline
Min $\mathcal{D}$ &  62 & 170 & 357 & 1758 & 7147 \\
\hline
$|V|$ in model Hole3 & 	340 & 660 & 1468 & 3136 & 6135\\
\hline
Max $\mathcal{E}$ & 9  & 18 & 38 & 99   &  817 \\
\hline
Min $\mathcal{D}$ & 22  & 72 & 187 & 760 &  2870\\
\hline
\end{tabular}
\end{center}
\caption{ The rows of Max $\mathcal{E}$ and Min $\mathcal{D}$ collect the timings (in seconds) of the first step and the second step in Algorithm 2
respectively.
\label{tbl:timing}
}
\end{table*}

\begin{figure}[!t]
\begin{center}
\begin{tabular}{ccc}
\includegraphics[width=0.3\textwidth]{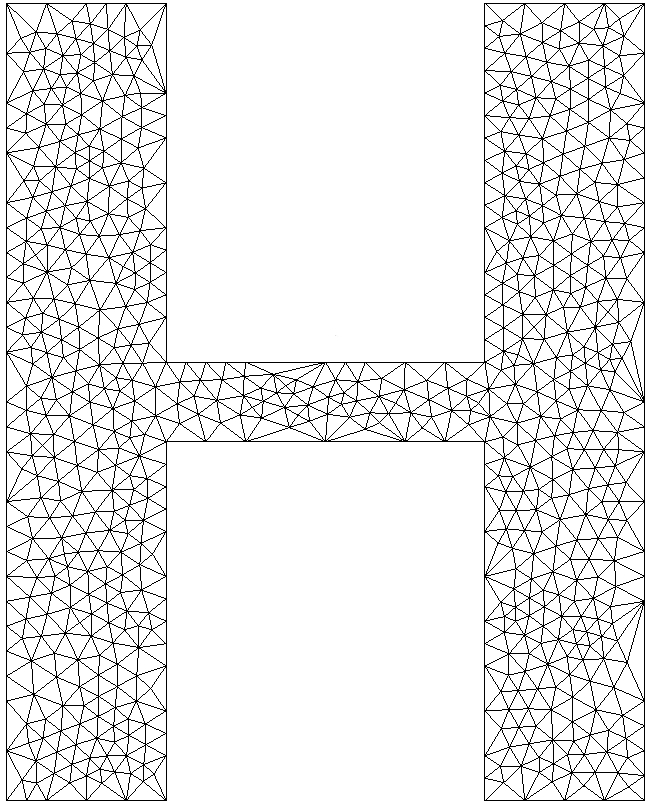} & \includegraphics[width=0.3\textwidth]{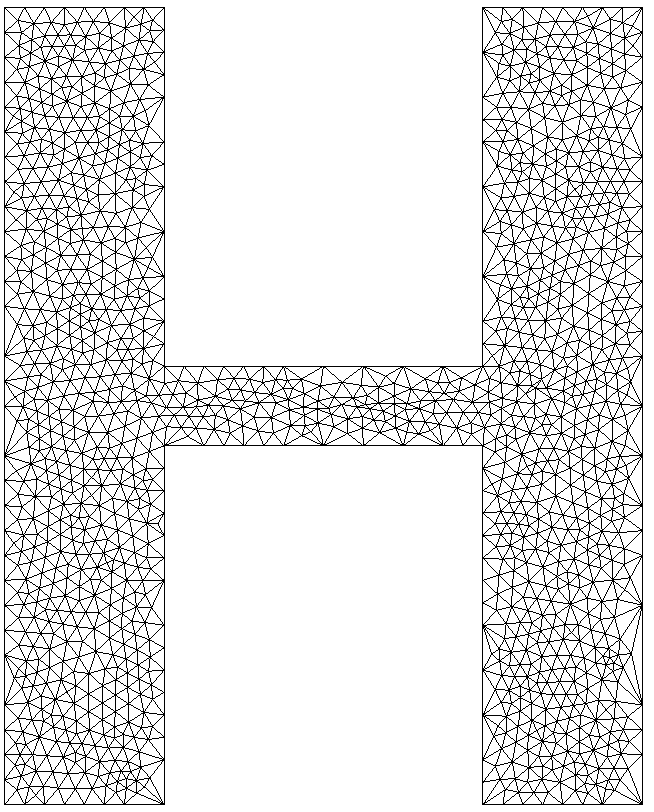} & \includegraphics[width=0.3\textwidth]{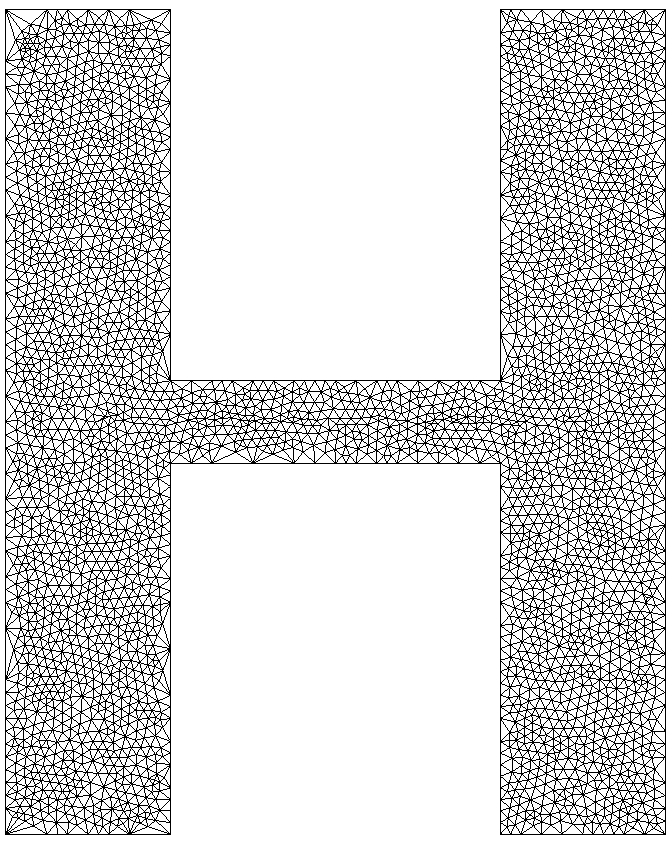}\\
\includegraphics[width=0.3\textwidth]{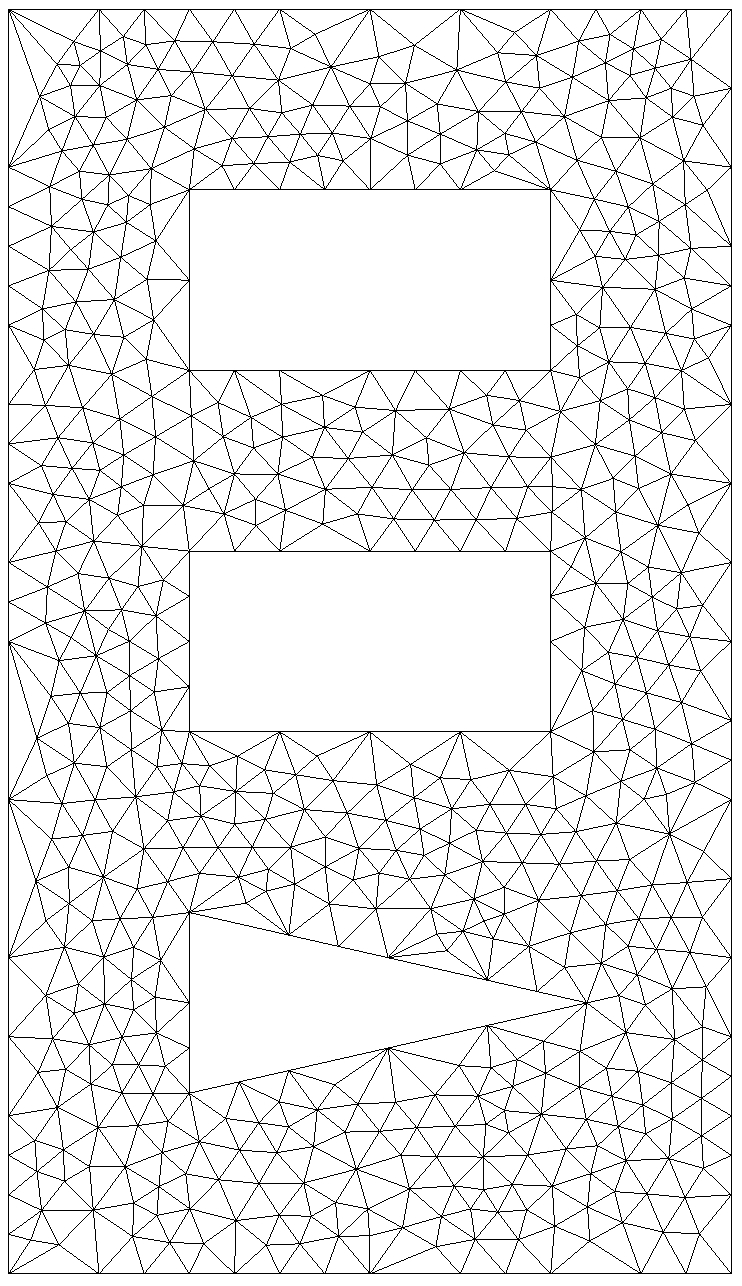} & \includegraphics[width=0.3\textwidth]{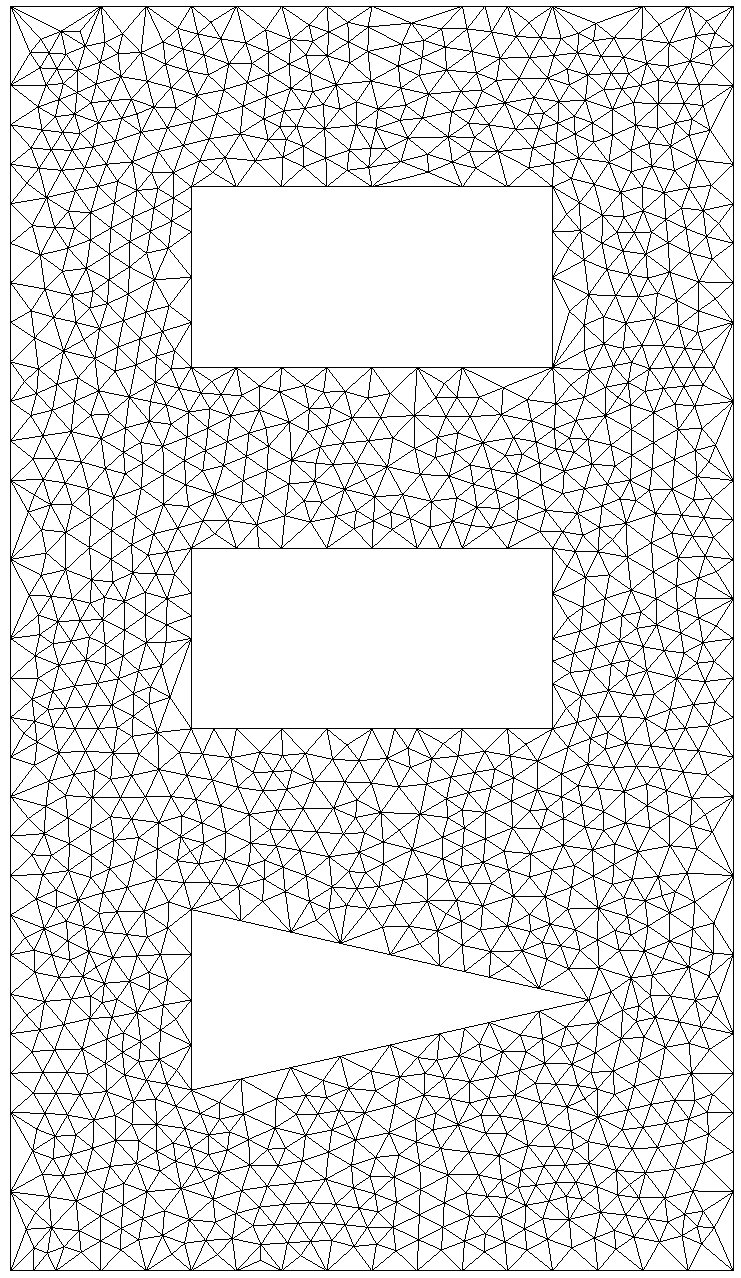} & \includegraphics[width=0.3\textwidth]{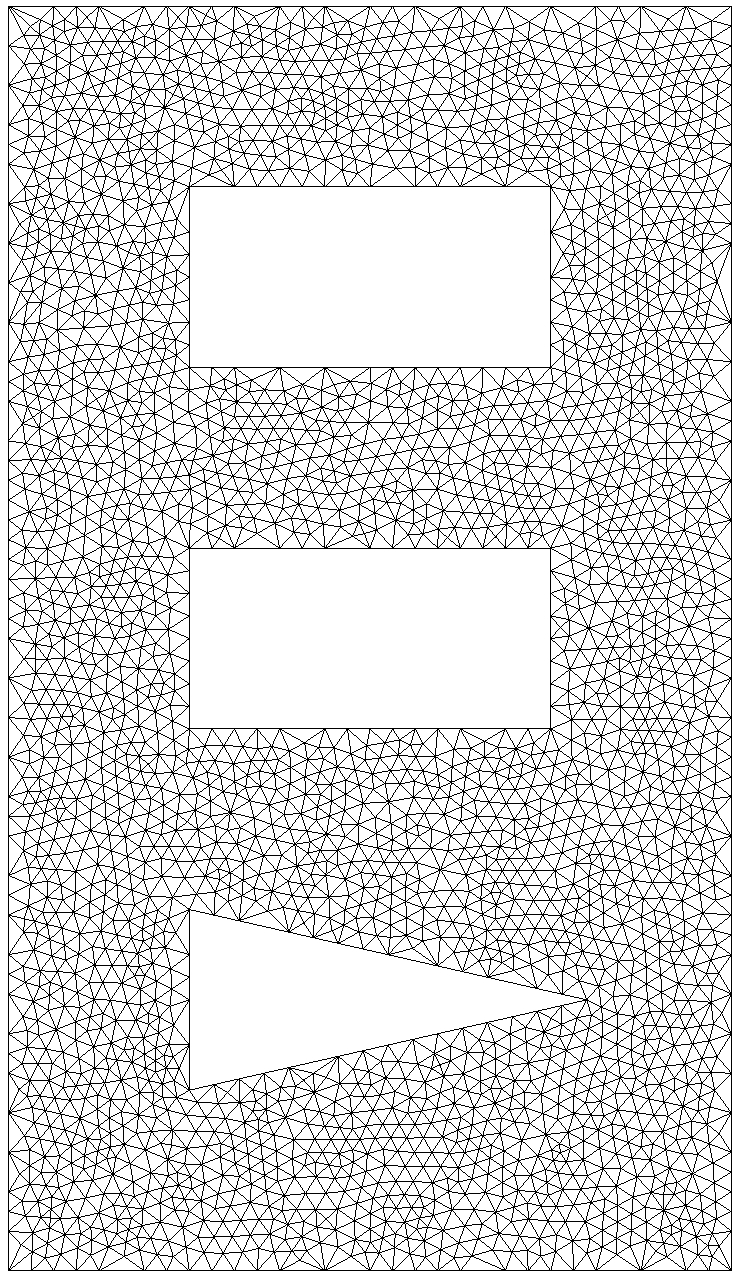}
\end{tabular}
\end{center}
\vspace{-4mm}
\caption{The first row shows H model with 756/1476/3584 vertices. 
The second row shows Hole3 model with 660/1468/3136 vertices. 
}
\label{fig:H_hole3}
\end{figure}


\begin{figure}[!t]
\begin{center}
\begin{tabular}{c}
\includegraphics[width=0.7\textwidth]{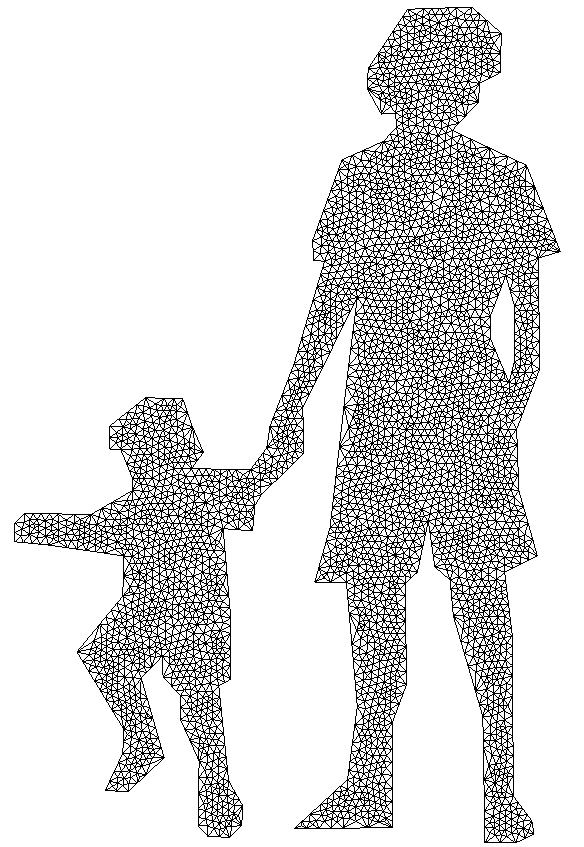}\\ 
\end{tabular}
\end{center}
\vspace{-4mm}
\caption{A mesh with 4007 vertices. 
}
\label{fig:family}
\end{figure}

\begin{figure}[!t]
\begin{center}
\begin{tabular}{c}
\includegraphics[width=0.95\textwidth]{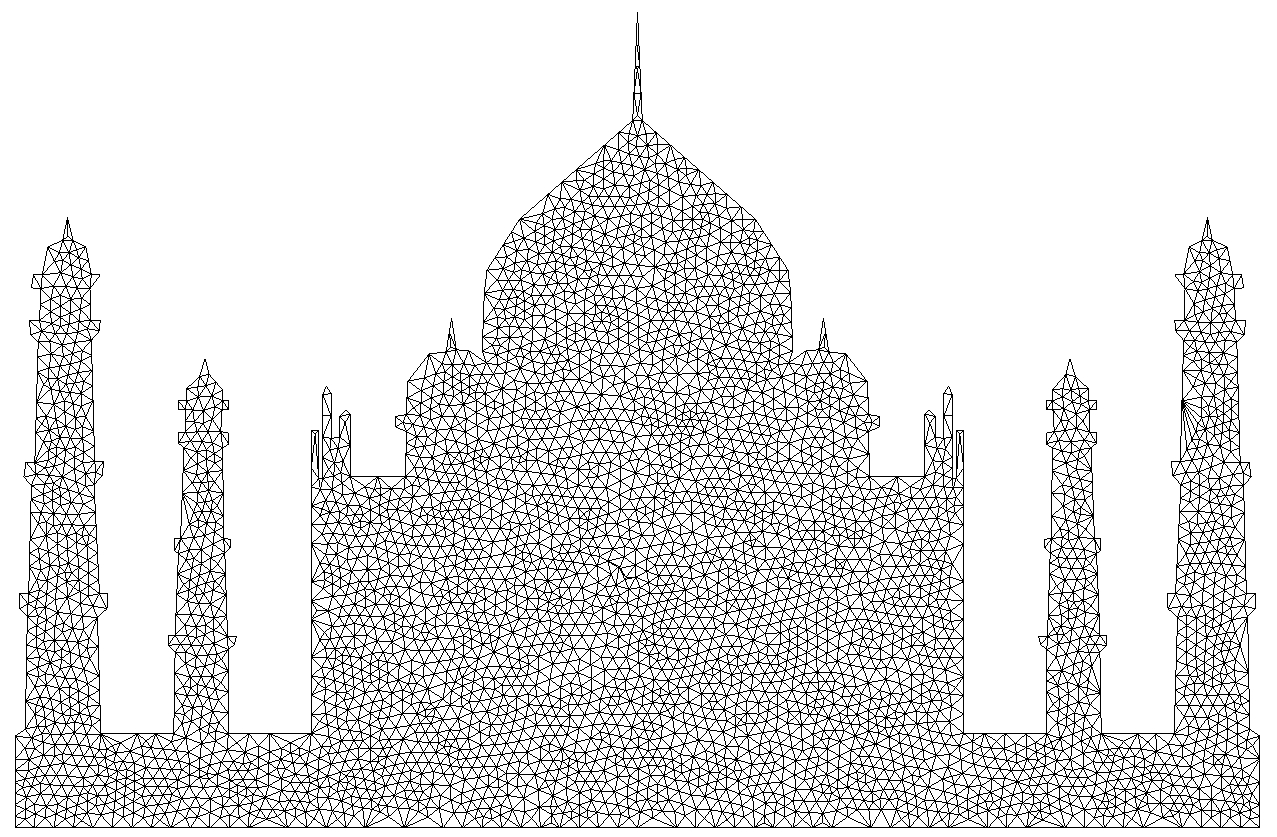}
\end{tabular}
\end{center}
\vspace{-4mm}
\caption{ A mesh with 4459 vertices. 
}
\label{fig:taj}
\end{figure}

\section{Discussion}
\label{sec:conclusion}
In this paper, we have proposed a new variational principle for improving 2D triangle meshes
based on hyperbolic volume, devised an efficient algorithm to maximize the 
energy functional over nonlinear constraints and to improve the quality of meshes, 
and applied our algorithm to various datasets and compared its performance to CVT. 
Here we point out a couple of possible directions for future work. First, notice that
the combinatorial structures of input meshes are fixed in the current framework.
However, one can also make them Delaunay \footnote{The sum of the opposite angles is 
less than $\pi$ for each interior edge} by flipping edges to improve the quality of meshes. 
It is interesting to see how the energy functional changes while flipping edges. 
Second, notice that the energy functional can be extended to higher dimensional spaces 
by considering higher dimensional ideal simplex. Thus one can follow the similar framework 
to improve higher dimensional meshes.   

\section{Acknowledgments}
The authors acknowledge The National Basic Research Program of China (973 Program 2012CB825501); Tsinghua National Laboratory for Information Science 
and Technology（TNList）Cross-discipline Foundation. 

%
 \bibliographystyle{plain}
 \bibliography{ref}
%


\printindex
\end{document}